\documentclass[a4paper]{article}
\usepackage[utf8]{inputenc}
\usepackage{marginnote}
\usepackage{todonotes}
\usepackage{titling}
\usepackage{fullpage}
\usepackage{subfiles, float, graphicx, amsmath, amssymb, amsthm, caption, enumitem ,fancyhdr, thmtools, thm-restate,  xspace, hyperref, mathtools}

\declaretheorem[style=definition]{definition}

\declaretheorem[style=plain,sibling=definition]{lemma}

\newcommand {\R}{\ensuremath{\mathbb{R}}}

\newcommand {\Z}{\ensuremath{\mathbb{Z}}}

\newcommand {\E}{\ensuremath{\mathbb{E}}}
\newcommand {\disk}{\ensuremath{\mbox{disk}}}
\newcommand {\dist}{\ensuremath{\mbox{dist}}}
\renewcommand {\line}{\ensuremath{\mbox{line}}}

\newcommand {\BIT}{\texttt{bit}}
\newcommand {\BITC}{\texttt{bit}^*}

\newcommand {\ER}{\ensuremath{\exists\mathbb{R}}}
\newcommand {\eps}{\ensuremath{\varepsilon}}
\newcommand {\NP}{\textrm{NP}}
\newcommand {\SmoothAna}{Smoothed Analysis}

\newcommand {\gridwidth}{grid width}
\date{\vspace{-5ex}}

\theoremstyle{plain}
\newtheorem{alphtheorem}{Theorem}

\title{\vspace{-4em}\SmoothAna{}  of Order Types}
\author{ Ivor van der Hoog\thanks{I.d.vanderhoog[at]uu.nl, t.miltzow[at]googlemail.com, M.C.vanschaik[at]students.uu.nl. \newline
 Utrecht University, Netherlands, The second author acknoledges generous support from the NWO Veni grant EAGER} \and Tillmann Miltzow\footnotemark[1]{} \and Martijn van Schaik\footnotemark[1]{}}

\begin{document}
\nocite{SchaikThesis}

\maketitle
\thispagestyle{empty}
\setcounter{page}{0}

\begin{abstract}
    Consider an ordered point set $P = (p_1,\ldots,p_n)$,
    its order type (denoted by $\chi_P$) is a map which assigns to 
    every triple of points a value in $\{+,-,0\}$ based on whether the
    points are collinear(0), oriented clockwise(-) or
    counter-clockwise(+).
    An abstract order type is a map
    $\chi : \left[\substack{n\\3}\right] \rightarrow \{+,-,0\}$ (where $\left[\substack{n\\3}\right]$ is the collection of all triples of a set of $n$ elements) that satisfies the following condition:
    for every set of five elements $S\subset [n]$ its induced order type $\chi_{|S}$ is \emph{realizable} by a point set.
    To be precise, a point set $P$ \emph{realizes} 
    an order type $\chi$,
    if $\chi_P(p_i,p_j,p_k) = \chi(i,j,k)$, for all $i<j<k$.
    
    Planar point sets are among the most basic and natural
    geometric objects of study in Discrete and Computational Geometry.
    Properties of point sets are relevant in theory and practice alike.
    An order type captures the structure of a point set. 
    In particular, many geometric properties of a point set (for example which
    points lie on the convex hull or the number of triangulations of that 
    point set)  depend only on its order type. 
    
    An interesting and well-studied question that surrounds order types, is which order types are realizable: for which abstract order types does there exist at least one (real-valued) point set that realizes that order type.
    A better understanding of realizability gives a better understanding of what kind of combinatorial and geometrical properties an arbitrary point set can have. 
    It is known, that deciding if an abstract order type is realizable is complete for the existential theory of the reals, i.e., $\exists\mathbb{R}$-complete~\cite{mnev1988universality, shor1991stretchability,richter1995mnev, DBLP:journals/corr/Matousek14}.
    Therefore, deciding if an abstract order type is realizable is \NP-hard, but it is unknown whether this problem is also \NP-complete. 
    
    Discrete realizable order types (order types that can be 
    realized by a point set with discrete rather than real-valued coordinates) 
    have also been extensively studied.
    For any abstract order type $\chi$ the \emph{norm} 
is the smallest grid $\Gamma = [0,M]^2\cap \Z$ that a 
point set $P\subseteq \Gamma$ requires to realize $\chi$.
If the norm of an abstract order type is a value $\nu$, then any point set that realizes that order type has at least one point with a coordinate that uses $\log(\nu)$ bits.
It is known that there exists classes of order types where 
the norm is $2^{2^{\Theta (n)}}$.
  Recently, Fabila-Monroy and Huemer ~\cite{RandomPoints}
  studied the realizability of the order type of random real-valued point sets.
  They showed that the order type of
  a random real-valued point set of $n$ points $P$ 
  has norm $n^{3+\varepsilon}$
  \emph{with high probability}.
  Independently, Devillers, Duchon, Glisse, and  Goaoc~\cite{xavier} 
  attained the same result. 
  Moreover they upper bound the expected number of bits
  that an algorithm needs to identify the order type of a random point set by $4n\log n + 16n$.

  Thus there are two extremal results regarding the realizibilty of the order 
  type of a real-valued point set of $n$ points:
  in the worst case the norm of their realized order type is doubly-exponential 
  in $n$, on the other hand with high probability a random point set 
  has a norm of at most $n^{3+\varepsilon}$.
  In this paper, we 
  study order type realizability under the lens of
  Smoothed Analysis 
  to interpolate between these two extremal results.
  We prove that if you randomly perturb an arbitrary point set with a perturbation
  of magnitude $\delta$, then with high probability the order type of 
  the perturbed point set has a norm of at most $\frac{1}{\delta}n^{3+\varepsilon}$. 
Our result implies the results from both~\cite{RandomPoints} as well 
as~\cite{xavier} with an arguably easier proof.
In addition, we also provide upper bounds for the \emph{expected} complexity of the grid width, 
the norm and the expected number of bits needed 
to describe the abstract order type realized by a perturbed real-valued point set.
In a nutshell, our results show that order type realizability is much
easier for realistic instances than in the worst case.
In particular, we can recognize instances in ``expected \NP-time''.
This is one of the first $\exists\mathbb{R}$-complete 
problems analyzed under the lens of Smoothed Analysis~\cite{ArxivSmoothedART}.
\end{abstract}

\section{Introduction}

\label{sec:intro}
We study the problem of order type
realizability and 
the computational complexity thereof. In this introduction, 
we first give a brief historic account on the problem and explain classical results regarding its
complexity.
Thereafter, we discuss the practical relevance of the 
problem and present recent developments in the average case analysis of the realizabilty of order types.
The juxtaposition between the classical negative worst case results and the 
recent positive average case results motivates us to study order type realizability 
under the lens of \SmoothAna{}. In this section we explain the general concept of \SmoothAna{} 
and elaborate on how we apply it in our case. Finally we present our findings
in Section~\ref{sec:results} on page~\pageref{sec:results}.
Sections~\ref{sec:WHP} and~\ref{sec:exp} contain the formal proofs. 
Section~\ref{sec:conclusion} concludes with a concise summary of the impact of our results.

\begin{figure}[htbp]
    \centering
    \includegraphics{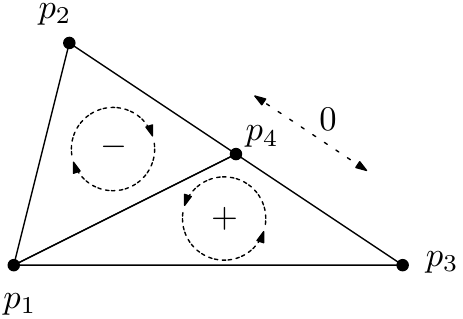}
    \caption{A configuration that describes the possible orientations of triples of points: the orientation of $\{p_1,p_2,p_4 \}$ is `-', that of $\{p_1,p_3,p_4\}$ is `+', and the collinear triple $\{p_2,p_3,p_4\}$ has orientation `0'.}
    \label{fig:orientation}
\end{figure}

\paragraph*{Definition of order types.}

    Consider an ordered point set $P = (p_1,\ldots,p_n)$.
    Its order type, denoted by $\chi_P$ is a map which assigns 
    to any triple of points a value in $\{+,-,0\}$ based on whether 
    the triple is oriented clockwise~(-), counter-clockwise~(+) or if they are collinear~(0).
    See Figure~\ref{fig:orientation} for an illustration. 
    An abstract order type is a map
    $\chi : \left[\substack{n\\3}\right] \rightarrow \{+,-,0\}$,
    where $\left[\substack{n\\3}\right]$ is the collection of all 
    triples of a set of $n$ elements, that satisfies the following condition:
    for every set of five elements their collective order type is \emph{realizable} by a point set.
    We say a point set $P$ \emph{realizes} 
    an order type $\chi$,
    if $\chi_P(p_i,p_j,p_k) = \chi(i,j,k)$, for all $i<j<k$.
    An order type is \emph{simple}, if there are 
    no collinear triples.
    Saying that a point set in the plane is  
    in \emph{general position} is equivalent to saying that its order type is simple.
    Throughout this paper we 
    denote by $P$ a set of $n$ real-valued, ordered planar points.
    Abstract order types are frequently referred to 
    as chirotopes, oriented matroids (of rank $3$) and pseudoline arrangements.
    There are books and chapters dedicated to those 
    topics~\cite{knuth1992axioms, richter20176,bjorner1999oriented, goodman2004pseudoline, felsner2017pseudoline, ziegler1996oriented}.
    \nocite{ExtremePilz,felsner1997number}
    
    In this paper, we consider the algorithmic question
    to realize an abstract order type by a 
    set of points in the plane.
    Assuming we would not know anything about abstract
    order types, we could wonder if all abstracts order
    types are realizable by a point set. The answer is no, and goes back to 
    Pappus, as we will discuss in the next paragraph.
    
\paragraph*{Pappus's hexagon theorem.}
Pappus of Alexandria lived in the fourth century 
and he
proved the following result, illustrated 
in Figure~\ref{fig:Pappus}.
\begin{alphtheorem}[Pappus’s hexagon theorem]
Let $p_1$, $p_2$, $p_3$ be three collinear points 
and let $p_4$, $p_5$ and $p_6$ also be three collinear points. 
The lines $\ell(p_1,p_6)$, $\ell(p_1,p_5)$,
$\ell(p_2,p_6)$ intersect the lines $\ell(p_4,p_3)$,
$\ell(p_4,p_2)$, $\ell(p_5,p_3)$, respectively, and
these three points of intersection are collinear.
\end{alphtheorem}
\begin{figure}[htbp]
    \centering
    \includegraphics[]{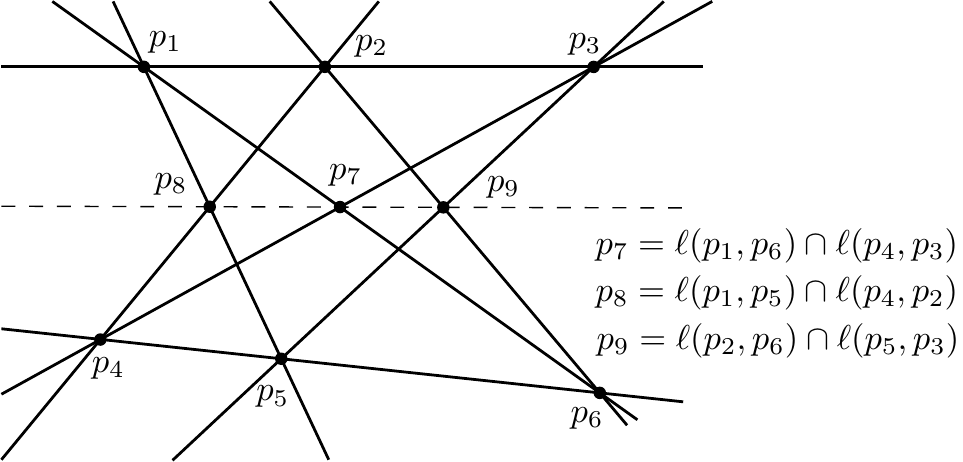}
    \caption{The configuration used in Pappus' hexagon theorem: 
    when $\{p_1,p_2,p_3\}$ and $\{p_4,p_5,p_6\}$ are two sets of collinear points, 
    then the points of intersection $\{p_7,p_8,p_9\}$ are collinear as well.}
    \label{fig:Pappus}
\end{figure}

\noindent
 Pappus's hexagon theorem describes a point set $Q$ of size nine that realizes a certain order type. 
 The construction assigns an order type to all triples apart from $\{p_7,p_8,p_9\}$ (the points of intersection) and shows that in that case, the three points of intersection must be collinear and therefore it must be that  $\chi_Q(p_7,p_8,p_9) = 0$. An abstract order type 
 however would be free to map the order type of $\{p_7,p_8,p_9\}$ to any value.
 Hence not every abstract order type can be realized by a point set in the plane. Having settled that not every order type is
realizable, a natural question is to wonder what
we know about realizable order types.
In particular whether the point set that realizes an order type can have integer coordinates.
This is the subject of the next paragraph.

\paragraph*{Norm of order types.}
For realizable abstract order types 
there are some order types that are more
easily realizable than others. 
For example, Gr\"unbaum has shown that configurations exist
that cannot be expressed with rational coordinates~\cite{Grunbaum1972}. 
In contrast, it is easy to see that 
\emph{simple} abstract order types can always 
be realized by integer coordinates. 
In other words simple order types can be 
represented as points on a discrete grid of a certain size. 
The norm $\nu(P)$ is the minimum size  grid  needed to describe a point 
set $P'$ that realizes the order type of $P$. 
Or formally:
\[\nu(P) = \mathrm{min}\; \mathrm{max} \{x'_1,x'_2, \dots ,x'_n,y'_1,y'_2,\dots,y'_n\},\] 
where the minimum is taken over all point sets 
$P' \subset \mathbb{N}_0^2$ that realize 
the order type of $P$.\footnote{Other papers (\cite{ExpCoord, RandomPoints, xavier}) defined the norm as the minimum over all $P' \subset \mathbb{Z}^2$. We chose to restrict $P'$ to $\mathbb{N}_0^2$ to have a straightforward relation between norm and bit-complexity. Our norm is at most twice the alternative norm.}
Goodman, Pollack and Sturmfels~\cite{ExpCoord}
show that the norm of a point can be
doubly exponentially large. In other words, 
for each value of $n$ there is a point set with norm at least $2^{2^{cn}}$, 
for some fixed constant~$c$.
 Complementing this lower bound, 
 they~\cite{ExpCoord} show that the norm 
 of a point set of $n$ points, in general position, 
 is upper bounded by $2^{2^{Cn}}$, for some constant $C$. 
 They used a mathematical connection between order types 
 and semi-algebraic sets, as explained below.
 This allowed them to use 
 a result defined in terms of 
 algebraic sets~\cite[Lemma~10]{Grigorev1988} 
 to find the maximal grid size needed to represent 
 any point set of $n$ points in general position. 

To summarize: if a point set $P$ has a norm of $\nu(P)$ then there exists 
a point set $P'$ with the same order type, 
where coordinates of $P'$ have at most $\nu(P)$ 
different values. So the number of bits required to express 
the largest coordinate of $P'$ in binary is at most $\log(\nu(P))$.
One might be tempted to try to find 
a different more concise
representation of order types of point sets.
To this end, we mention a deep connection
between order type realizability and real algebraic geometry.
As we will see this connection is a strong argument (but not a proof)
against the existence of potentially more concise representations.
\paragraph*{Mn\"{e}v's universality theorem.}
    The connection between semi-algebraic sets and order types is deep
    and the upper bounds mentioned in the previous paragraph was only
    the tip of the iceberg. 

    One of the most astounding results regarding order types, 
    is by Mn\"{e}v~\cite{mnev1988universality}. 
    In order to describe it properly, 
    we need to introduce semi-algebraic sets first.
    Given a set of polynomials 
    $P_1,\ldots,P_n \in \Z[X_1,\ldots,X_k]$ and another set
    $Q_1,\ldots,Q_m\in \Z[X_1,\ldots,X_k]$, we define the set 
    \[S := \{x\in \R^k : P_i(x) = 0, \forall i; \ Q_j(x) > 0, \forall j \}.\]
    Sets that can be defined in this way are \emph{semi-algebraic} sets.
    Semi-algebraic sets are versatile, as many practical
    problems can be described precisely with the help of polynomials.
    E.g. for any three points $p,q,r \in \R^2$, we can
    define its orientation as the sign of a polynomial using six variables:
    \[
    \det \left(\begin{array}{ccc}
        p_x & q_x & r_x \\
        p_y & q_y & r_y \\
        1 & 1 & 1
    \end{array}\right).
    \]
    Thus, given an abstract order type $\chi$, we can construct 
    a semi-algebraic set $S_\chi$, such that the points $S$ 
    and the points in the plane realizing $S_\chi$ are in one-to-one
    correspondence. This connection was used by~\cite{ExpCoord}, to
    show the upper bound on the norm of a point set as discussed above.
    The interesting part is that this connection also works the other
    way around. This result is widely known as 
    Mn\"{e}v's universality theorem.
    To explain this, we define the realization space
    of an abstract order type:
    \[R(\chi) = \{x\in \R^{2n}: x_i \in [0,1], x \text{ as a point set in the plane realizes }\chi\}.\]
    
    Given a semi-algebraic set $S$, we can construct a
    \emph{simple} abstract order type $\chi_S$, such that 
    $S \simeq R(\chi_S)$. 
    Here the symbol $\simeq$ denotes stably-equivalence. 
    As the definition of
    stably equivalent is fairly technical, let us just highlight key 
    properties of stably equivalent sets. 
    If two sets are stably equivalent, then they may have different
    intrinsic dimensions, but the number of connected components is the
    same and elements of one set can be transformed ``easily''
    to elements of the other set.
    In particular if $S \simeq T$ then 
    $S = \emptyset \Leftrightarrow T = \emptyset$.
    This shows that the anecdotal results that we saw before, i.e.,
    Pappus hexagon theorem and the doubly exponential lower bound 
    found by Goodman, Pollack and Sturmfels are not just 
    isolated phenomena, but actually stem from a deeper 
    mathematical connection to real algebraic geometry.
    The existence of this connection also diminishes hopes to find a more
    concise description of the realization space of 
    an order type or even of the point sets contained in them.
    In particular, it is not known, nor believed that
    we can test if a semi-algebraic set is empty
    in non-deterministic polynomial time. In other words,
    order type realizability is likely not contained
    in the complexity class \NP.
    Instead, it is complete for the complexity class 
    called the existential theory of the reals or \ER, for short.
    In particular, a short description of a realization of 
    an abstract order type would
    imply $\NP = \ER$.
    In the next paragraph, we discuss this complexity class.

\paragraph*{Existential theory of the reals.}
    The complexity class \ER{} captures all problems
    that are under polynomial time reductions equivalent
    to the problem of deciding if a given semi-algebraic set is 
    empty or not.
    By now we have a long list of problems 
    that are complete for this complexity class~\cite{AnnaPreparation,AreasKleist,shitov2016universality, richter1995realization,garg2015etr,schaefer2013realizability, cardinal2017intersection, cardinal2017recognition, kang2011sphere,ARTETR, herrmann2013satisfiability, TensorRank, NashSchaefer}.
    Since there is a polynomial-time reduction between these problems,  
    an \NP{} algorithm that works for one of them,
    implies the existence of an \NP{} 
    algorithm for all of them.
    Although it is not certain whether \NP{} 
    $\not =$ \ER{}, \ER-completeness represents 
    our current barrier of knowledge.
    To continue to make progress on the algorithmic 
    problems that are \ER-complete, we need to relax the
    problem and there are various ways of doing that.
    One of them is to consider average case analysis.

    \subsection{Average case analysis and prior average case results}
Recently two groups of researchers have studied 
order type realizability for the \emph{average case}.
The core idea of average case analysis is that we draw an instance
(according to some given distribution) among all the 
instances and then we consider the expected costs of the algorithm
according to this distribution.
Several remarks are in order:
First, the costs of an algorithm is usually the running time.
But in our case, we are interested in the norm of an order type,
or maybe its bit-complexity.
The results that we will attain depend on the precise
definition of the cost function.
Second, there are various ways to choose the distribution 
over which you take the average. 
Different distributions may lead to different
results. 
Third, with many algorithmic problems 
a certain type of instance is often dominating
the analysis. 
Average case analysis 
could 
under-represent problematic instances, even though 
they resemble typical instance found in practice. 
Those are usually considered strong argument 
against the relevance of average case analysis
to explain practical performance of algorithms.
Fourth, with average case analysis you typically 
see one of two types of results: either we  
bound the \emph{expected costs}; 
or we show that low costs appear
\emph{with high probability}.
The first type of statement is stronger for upper bounds, as
it takes into account events that are rare, 
but may also have 
very high costs.

\paragraph*{Prior average case results.}

With these remarks in mind, we want to point out first
the result by Fabila-Monroy and Huemer~(Theorem 1, \cite{RandomPoints}).
They consider $n$ real-valued points drawn
independently and uniformly at random from the unit 
square $[0,1]^2$. They show that with high probability 
the norm of the point set is upper bounded by
$\lfloor n^{3+\epsilon}\rfloor$.
Note that order types are scale-invariant and therefore any realizable order type can also be realized within the unit square.
This polynomial upper bound on the norm of a real-valued point set is a 
huge improvement over the double exponential 
upper bound given by~\cite{ExpCoord}.
However the result comes with all the caveats mentioned above.

Devillers, Duchon, Glisse,and Goaoc~(Theorem 3, \cite{xavier}) reproved 
this result independently.
Furthermore, they showed that there exists an algorithm that can determine 
the order type of a random point set by reading at most $4n \log n + 16n$ coordinate bits.
In a restricted model of computation 
at least $4n \log n - 4n \log \log n$ coordinate bits are required,
for the same distribution on the order types.
This result is considerably stronger as it considers
the number of bits as the cost function and makes 
a statement about the expected costs rather than
a statement about high probability of positive events.
In other words, it takes into consideration events 
with low probability but potentially high costs.
As those are encouraging positive results, we are motivated
to consider a more fine grained form of analysis, which
may overcome the previously mentioned drawbacks.
We will discuss this in the next subsection.

\subsection{Smoothed Analysis and measuring costs}
\label{par:Smooth}
Spielman and Teng~\cite{spielman2004smoothed} proposed a new analysis called 
\emph{\SmoothAna{}}, where the performance of an algorithm 
is studied under slight perturbations of arbitrary inputs. Intuitively, smoothed analysis interpolates between average case and worst case results.
They explain their analysis by applying it on the Simplex algorithm,
which was known for particularly good performance in practice 
that was impossible to verify theoretically~\cite{klee1970good}. 
Using this new approach, they show that the Simplex algorithm 
has ``smoothed complexity'' polynomial in the input size and 
the standard deviation of Gaussian perturbations of those inputs, 
which was the desired theoretical verification of its good performance.
See Dadush and Huiberts~\cite{dadush2018friendly}, 
for the currently best analysis.

\paragraph*{Prior applications of \SmoothAna{}.}
Since its introduction, \SmoothAna{} has been applied to numerous 
algorithmic problems.
For example the \SmoothAna{} of the Nemhauser-Ullmann
Algorithm~\cite{nemhauser1969discrete} for the knapsack problem shows
that it runs in smoothed polynomial time~\cite{KnapsackSmooth}. A more general
result that was obtained using \SmoothAna{} is the following: all binary optimization problems (in fact, even a larger class of combinatorial problems) can be solved in smoothed polynomial time if and only if they can be solved in
pseudopolynomial time~\cite{beier2006typical}. 
Other famous examples are
the \SmoothAna{} of k-means algorithm~\cite{arthur2006worst}, the
2-OPT TSP local search algorithm~\cite{englert2007worst}, and the local
search algorithm for MaxCut~\cite{MaxCUTsmoothed}. Not surprisingly,
teaching material on this subject has become
available~\cite{smoothCourseHeiko,smoothCourseTim,SurveySmoothTim}.
Most relevant for us is the \SmoothAna{} of the 
Art Gallery Problem~\cite{ArxivSmoothedART}, as this is
another \ER-complete problem. Roughly speaking, 
the authors showed that the Art Gallery Problem
can be solved in ``expected \NP-time'', 
under the lens of \SmoothAna{}.
This paper is the second time that \SmoothAna{} is applied to an \ER-complete problem and we show that order type realizability can be solved in ``expected \NP-time'', 
under the lens of \SmoothAna{}.

\paragraph*{Formal definition of \SmoothAna{}.}
In this paragraph, we will formally define the smoothed
complexity of an algorithm.
Let us fix some $\delta$, 
which describes the \emph{magnitude of perturbation}.
The variable $\delta$ describes by how much we allow to perturb the original input. In this paper we consider as input ordered planar point sets and we perturb the original input by replacing each point with a new point that lies within a distance $\delta$ of the original.
We denote by ($\Omega_\delta$,$\mu_\delta$) 
the probability space where each $x \in \Omega_\delta$
defines for each instance $I$ a new `perturbed' instance $I_x$. 
We denote by $\mathcal{C}(I_x)$ 
the cost of instance $I_x$.
The smoothed expected cost of instance $I$ equals:
\[ \mathcal{C}_\delta(I) = \underset{x\in\Omega_\delta}{\E}  \mathcal{C}(I_x) = \int_{\Omega_\delta} \mathcal{C}(I_x)\mu_\delta(x) \; \mathrm{d}x. \]
If we denote by $\Gamma_n$ the set of all instances of size $n$,
then the smoothed complexity equals:
\[ \mathcal{C}_{\textrm{smooth}}(n,\delta) = \max_{I\in\Gamma_n}\;\underset{x\in\Omega_\delta}{\E}\left[ \mathcal{C}(I_x) \right].\]
This formalizes the intuition mentioned before: 
not only do the majority of instances behave nicely, 
but actually in every neighborhood 
(bounded by the maximal perturbation $\delta$) the majority of 
instances behave nicely. The smoothed complexity is measured in 
terms of $n$ and $\delta$. If the expected complexity is small in 
terms of $1/\delta$ then we have a theoretical verification of the hypothesis 
that worst case examples are well-spread.
Before we explain the model of perturbation and the cost function
that we consider, we will explain the algorithm that we are using in the analysis.

\paragraph*{Snapping, the naive algorithm.}
Given a planar point set $P$ in the unit square it is non-trivial to determine 
its norm exactly. However, a simple way to get an upper bound 
is to \emph{snap} every point to a point onto a fine grid.
In other words, we fix a grid $\Gamma = w\mathbb{N}_0^2$ and for every
point $p \in P$ we denote by $p'\in \Gamma$  the closest grid point to
$p$. In this way, we attain a snapped point set $P'$. If we scale 
$P'$ by a factor of $1/w$ then we get an upper bound of the norm of $P$.

\subsection{Model of perturbation, cost functions and practical relevance}
Defining a  perturbation of an order type is non-trivial,
as it is a combinatorial structure with many dependencies.
A possible model is to take a random 
$\delta$-fraction of all the triples and randomly
reassign their orientation. However, it is likely that
we will get a map that is not an  abstract order type. Even if the resulting map is an abstract order type, it is then also highly likely that the order type is not realizable.
Hence we take a different approach and perturb a realization
of the order type as opposed to the order type itself (in other words, we are perturbing point sets).
This guarantees that the resulting order
type is realizable. 
Furthermore, we can upper bound the norm of such 
perturbed point sets via snapping, as described in the previous 
paragraph. 
Observe that a uniform distribution over a real-valued point set implies 
a distribution of the order type that the point set realizes. 
This distribution over the order type however does not have to be uniform.

Let us consider a set of $n$ real-valued points in the plane. In order to make
the magnitude of perturbation meaningful, we normalize the point
set and translate it to the unit square $[0,1]^2$ without changing its order type.
We define the perturbation space 
$\Omega_\delta = \disk(\delta)^n\subset \R^{2n}$.
Here $\disk(\delta)$ denotes a disk with radius $\delta$ around
the origin.
Given a specific ordered point set $P = (p_1,\ldots,p_n)\in \R^{2n}$
and a perturbation $x = (x_1,\ldots,x_n)\in \Omega_\delta$, 
we define the perturbed 
point set $P_x$ as 
$P_x = P + x = (p_1 + (x_1,x_2),\ldots,p_n+(x_{n-1},x_n))$.
Figure~\ref{fig:smoothed} shows an example. For each point $p$, 
its perturbed point $p_x$ lies within a distance 
of $\delta$ but its snapped point might not.
We briefly want to show that this perturbation space $\Omega_\delta$ also defines a 
perturbation space over order types: let $\chi_P$ be an order type realized 
by some point set $P$. Every $x \in \Omega_\delta$ gives a new 
point set $P_x$ with possibly a different order type $\chi_{P_x}$. 
So $\Omega_\delta$ applied to $P$, indirectly defines a distribution 
on the order types, given by $\Pr(\chi_{P_x} = \xi)$. 
Our results hold for all fixed choices of $P$. 

\begin{figure}[htbp]
    \centering
    \includegraphics[]{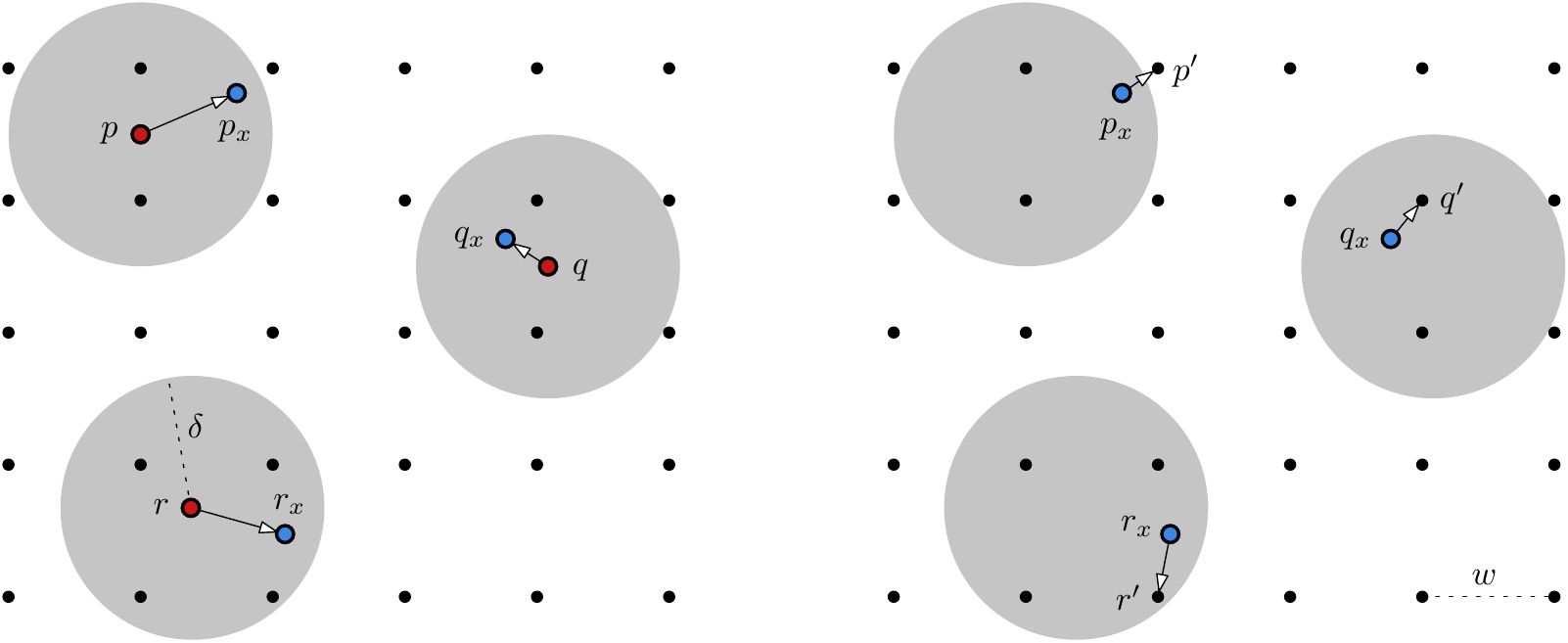}
    \caption{The instance $P=\{p,q,r\}$ 
    is randomly perturbed to $P_x = \{p_x,q_x,r_x\}$.
    The set $P_x$ is snapped to the points $P' = \{p',q',r'\}$.}
    \label{fig:smoothed}
\end{figure}

\paragraph*{Cost functions.}
There are three natural ways to define the ``cost'' of realizing the order type of a real-valued point set $P$.
The first is the grid width $w = w(P)$, which is the largest value 
$w \in [0,1]$, such that snapping $P$ onto a grid of width $w$ 
preserves the order type of $P$.
The second is the norm $\nu = \nu(P)$. 
If we perturb a point set $P$ by $\delta \in [0,\frac{1}{2}]$ then the 
resulting point set (when translated) lies within $[0,2]$. 
It is easy to see that the norm is at most inversely proportional to the grid width:
\[\nu \leq \lfloor 2/w \rfloor.\]
The third cost measure is the bit-complexity $\BIT(P)$: the minimal number of bits
needed to express the coordinates of a point set $P'$ that realizes the 
order type of $P$. $\BIT(P)$ is upper bound by:
\[ \BIT(P) \leq  2 n \log \nu(P) \leq  2 n \log(  \lfloor 2/w \rfloor ) \leq 2n(1 - \log w ),\]

as each of the $2n$ coordinates of $P'$ need at most $\log \nu$ bits.
Later in the paper, we observe that upper bounding the bit-complexity of every coordinate by $\log \nu(P)$ is a too pessimistic upper bound and that a more fine-grained analysis of the total bit-complexity yields better results.

The prior results that 
we mentioned~\cite{xavier, ExpCoord, RandomPoints}
mostly measured cost in the norm. 
These results were either worst case results or results 
with high probability 
and therefore any result measured in the norm translates into a result for 
the grid width and the bit-complexity using the observations we presented 
above. However when analyzing the expected cost of the realization, 
this is no longer true. This is because $f(\mathbb{E}(X)) \neq \mathbb{E}(f(X))$, for most $f$. 
This is why we analyze the expected cost separately for all three cost measures.
And we get different results for all three of them.

\paragraph*{Practical relevance.}
So far we talked about the theoretical importance of order type 
realizability, its worst case complexity and average case 
complexity. This motivated us to study the problem under the lens
of \SmoothAna{}. The main motivation of \SmoothAna{} 
is to give a realistic estimate of 
the practical performance of algorithms. Order type
realizability is fundamentally a theoretical problem. Nevertheless, it has also
some practical aspects: in Computational Geometry, many algorithms 
have implementations available, which can be found in the CGAL library.
Most of theses algorithms compute points in the plane and their
precision is important in those applications. Rounding is a common
source of incorrect code and CGAL saves the coordinates of points
in an exact manner. We have to maintain the
coordinates with such high precision in order to preserve 
the order type, to get consistent results from the algorithms.
Thus one can argue that a better understanding of the order type
realizability may lead to a better understanding of when to store precise coordinates 
and when rounding can be acceptable to speed up performance.

\subsection{Results}
\label{sec:results}
Our first result is that \emph{with high probability},
under small perturbations, a point set has  much lower costs
than the worst case suggests.  This gives us the following theorem:

\begin{restatable}{theorem}{WHP}\label{thm:WHP}
  Given $n$ points $P \subset [0,1]^2$ and a magnitude of 
  perturbation $\delta\leq 1/2$. Then it holds with probability 
  at least $1 - p$ that
  \begin{enumerate}[itemsep=0mm, label = (\arabic*)]
    \item the snapped point set $P'$ onto the grid with width 
    $w = \frac{p \delta}{ n^3}$ has the same order type,
    \item the norm of $P'$ is at  most 
    $\left \lfloor \frac{2 n^3}{ p \delta} \right \rfloor$, and
    \item at most $2n\log \left(   \frac{n^3}{p \delta} \right) +2n $ bits are needed to represent the point set.
  \end{enumerate}
\end{restatable}

\noindent
From Theorem~\ref{thm:WHP}, we can conclude a statement
about the average case complexity, by setting $\delta  = 1/2$
and consider the case that $P$ has all its points at the origin.
Note that Corollary~\ref{cor:Average} implies almost Theorem~$1$ by 
Fabila-Monroy and Huemer~\cite{RandomPoints} and Theorem~$3$  
by Devillers, Duchon, Glisse, and Goaoc~\cite{xavier} by substituting $p = n^{-\eps}$.
The only difference is that their point set lies in the unit square
and ours lie in the unit disk.

\begin{restatable}{corollary}{AverageCase}\label{cor:Average}
    Given $n$ points $P \subset \disk(1/2)$ chosen uniformly
    and independently at random.
    Then it holds with probability 
    at least $1 - p$ that
    \begin{enumerate}[itemsep=0mm,label = (\arabic*)]
        \item the snapped point set $P'$ onto the grid with \gridwidth{} $w = \frac{2p}{n^3}$ has the same order type,
        \item the norm of $P'$ is at most $\left \lfloor \frac{4 n^3}{p} \right \rfloor$, and
        \item at most $2n \log \left(   \frac{n^3}{p} \right) +2n$ bits are needed to the pointset.
    \end{enumerate}
\end{restatable}

\noindent
The more desirable result of this paper is a statement about the expected costs.
Integrating over the probabilities given in Theorem~\ref{thm:WHP}
gives us the following theorem:

\begin{restatable}{theorem}{Expected}\label{thm:Expected}
    Given $n\geq 2$ points $P \subset [0,1]^2$ and the magnitude of 
    perturbation is given by $\delta\le 1/2$. Then it holds that
    \begin{enumerate}[itemsep=0mm,label = (\arabic*)]
        \item \label{itm:EXPWIDTH}
        the expected required \gridwidth{}  
        is at least $\frac{\delta}{4 n^3} = \Omega \left(\frac{\delta}{n^3} \right)$.
        \item \label{itm:EXPNORM}
        the expected norm of $P'$ is upper bound by $\frac{n^3}{ \delta} {2^{cn}}$, for some constant $c$, and
        \item \label{itm:EXPBITS} the expected bit-complexity per coordinate 
        is upper bound by $\log \frac{ n^2 }{\delta} + 2$.
    \end{enumerate}
\end{restatable}

\noindent
The expected number of bits that a single coordinate needs to express the 
order type of $P$ is upper bound by $(\log(n^2 / \delta) + 2)$. 
Using linearity of expectation, this means that the expected number of bits 
needed to represent $n$ points is upper bound by 
\[2n\log(n^2/\delta)    + 4n = 4n\log{n}    + (4+2\log{1}/{\delta})n .\]
Considering the case $\delta = 1/2$, we have essentially reproven
Theorem~2~(ii) of~\cite{xavier} with the same leading constant 
and an improved constant in the lower order term ($16 \rightarrow 6$).
At this point we would like to mention that we consider the proofs 
for our more general results to be simpler than the proofs 
for the results of~\cite{xavier, RandomPoints}. 
There is of course no objective measure for the simplicity of a technique, 
however the results of \cite{RandomPoints} and \cite{xavier} were obtained using 
involved geometrical constructions and extensive probabilistic analysis. 
Our theorems however result from  one natural observation about order type 
preserving triangles and repeated application of the union bound. 
We hope simplifying the analysis for random order types contributes to 
further research into the realizable order types. 
Skip ahead to Section~\ref{sec:conclusion} for a concise overview of 
the impact of these results.

\paragraph*{Proof outline.}
We introduce the notion of $w$-flatness for
a triple of points. Roughly speaking $w$-flatness indicates how
close to collinear three points are. 
It is easy to upper bound the probability that a single triple
is $w$-flat. Using the union bound, we upper bound
the probability that any triple in $P_x$ is $w$-flat.
We show that if $P_x$
has no $w$-flat triple, then snapping it onto a grid of width $w$ preserves its
order type. 
The other results are obtained by applying standard probability theory.

The crux of our analysis is that we show that the average case grid width $w$ is 
a factor of $n^{-3}$ smaller than the perturbation magnitude $\delta$. 
This relation means that as the perturbation magnitude becomes smaller 
(as we get closer and closer to worst case analysis) the required \gridwidth{}
only decreases at a polynomial pace. Observe that even if the magnitude of 
perturbation is as small as $\delta = n^{-10}$ 
(i.e.\, the point set has a high precision) 
then the average case \gridwidth{} is  still 
significantly (doubly-exponentially) larger than the worst case \gridwidth{}~of~$2^{-2^{cn}}$. 
 
\section{Low Costs with High Probability}\label{sec:WHP}

In the previous section, we explained that we are given a specific 
ordered point set $P = (p_1,\ldots,p_n)\in \R^{2n}$ in the unit square
which is subject to a random perturbation of magnitude $\delta \le \frac{1}{2}$. 
\begin{figure}[h]
    \centering
    \includegraphics{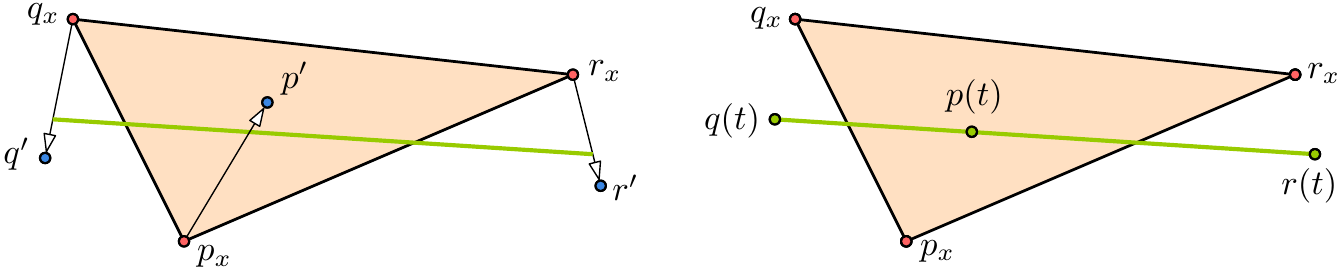}
    \caption{An illustration of the proof of Lemma~\ref{lemma:flat}. The green segment indicates the location of $\ell$.}
    \label{fig:lemmafour}
\end{figure}

We denote the perturbation by $x = (x_1,\ldots,x_n)\in \Omega_\delta$ and 
we denote the point set after the perturbation by 
$P_x = (p_1 + (x_1,x_2),\ldots,p_n+(x_{n-1},x_n))$ which we translate to 
lie inside $[0,2]^2$.
The perturbed point set $P_x$ consists of real-valued points, 
which we will snap onto a grid $\Gamma$ with a width $w \le 1$ to obtain 
a point set denoted by $P'$. 
We would like the order types of $P_x$ and $P'$ to be the same. 
To this end we call any triple $\{ p,q,r \}$ of points $w$-\emph{flat} 
if there is at least one point in the triple which lies within distance 
$\sqrt{2}w$ of the line through the other two and we prove the following lemma:

\begin{lemma}
\label{lemma:flat}
Let $p_x,q_x,r_x \in P_x$ be three perturbed points. If $\{p_x, q_x, r_x \}$ is not $w$-flat then the order type of $\{p_x, q_x, r_x \}$ and their snapped points $\{p', q', r' \}$ are the same.
\end{lemma}

\begin{proof}
This proof is illustrated by Figure~\ref{fig:lemmafour}.
We denote by $\dist(a,b)$ the Euclidean distance between points 
$a$ and $b$ and we denote by $\line(a,b)$ 
the line spanned by those points.
Consider the linear transformation between $\{p_x, q_x, r_x \}$ and 
$\{p', q', r' \}$ during a time interval $[0,1]$.  
If the orientation of $\{p_x, q_x, r_x \}$ and $\{p', q', r' \}$ 
are different then there must be a time $t \in [0,1]$ where the 
three points are collinear on a line $\ell$ and assume without 
loss of generality that $p(t)$ lies in between $r(t)$ and $q(t)$.
Since $p'$ is the result of snapping $p_x$ into a grid of 
width $w$ it holds that $\dist(p_x, p') \le {w} / {\sqrt{2}}$ and 
therefore $\dist(p_x, p(t)) \le {w} / {\sqrt{2}}$.
Similarly it must be that $\dist(q_x, q(t))$ and 
$\dist(r_x, r(t))$ are upper bound by~${w} / {\sqrt{2}}$.
However, since $p(t)$ lies in between $q(t)$ and $r(t)$, 
this implies that $\dist( p(t), \line(q_x, r_x))$ is 
at most~${w} / {\sqrt{2}}$.
Using the triangle inequality, the distance between $p_x$ 
and $\line(q_x, r_x)$ is at most 
$({w} / {\sqrt{2}} + {w} / {\sqrt{2}}) = \sqrt{2}w$ 
which implies that $\{p_x, q_x, r_x \}$ is $w$-flat 
and this proves the lemma.
\end{proof}

\begin{figure}[htbp]
    \centering
    \includegraphics{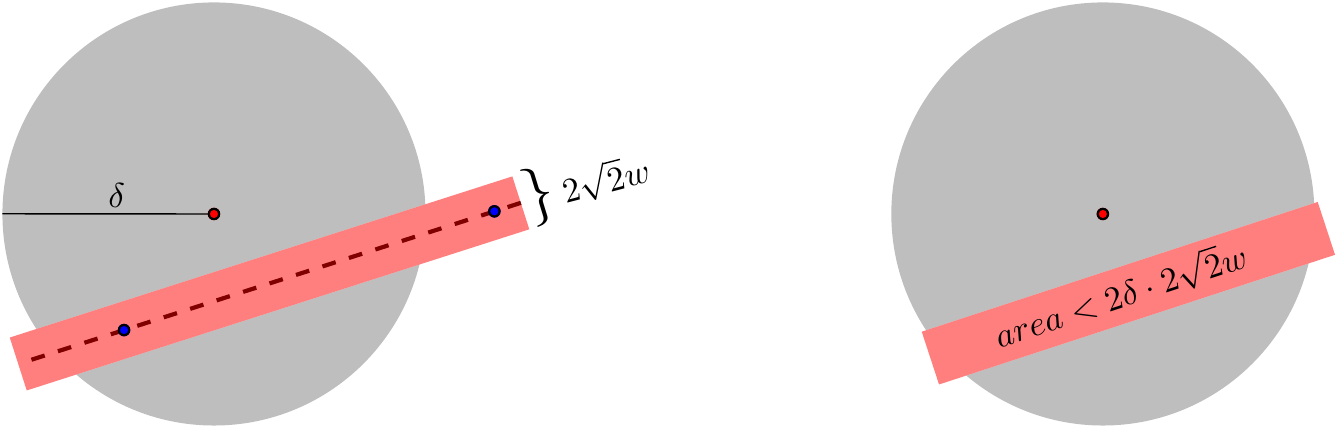}
    \caption{The intersection area of the perturbation disk and the area induced by the other two points is upper bounded by $2\delta \cdot 2\sqrt{2}w$.}
    \label{fig:BadAreas}
\end{figure}

\WHP*

\begin{proof}
This proof is illustrated by Figure~\ref{fig:BadAreas}.
Consider two perturbed points $p_x$ and $q_x$ and the area of all points
that lie within distance $2w$ of the line through them. For any other 
point $r_x$, the probability that $r_x$ is contained within this area is
upper bound by $(4\sqrt{2}w) / (\pi \delta) \leq (2w)/ (\delta)$. 
This is because the area of
intersection between this area and any disk of radius $\delta$ is 
less than $(2\sqrt{2}w)\cdot(2\delta)$ whereas the total area 
of a perturbation disk with radius $\delta$ is $\pi \delta^2$.

We define for $s \in \{p,q,r\}$ the event $A_s$ as the
event that the  $s$ is within distance $\sqrt{2} w$
with respect to the line through the other two points. 
The calculation above shows $\Pr(A_s)\leq {2w} / {\delta}$.
And this implies:
\[\Pr(A_p\cup A_q \cup A_r) \leq \Pr(A_p) + \Pr(A_q) + \Pr(A_r) \leq \frac{6w}{\delta}.\]
Therefore the probability that
any three points $ p_x, q_x, r_x  \in P_x$ are $w$-flat is upper 
bound by $6w / \delta$.

For every triple $T$ of points $T\subset P$,
we define the event $A_T$, that those three points are $w$-flat.
Using the union bound, we get
\[\Pr\left(\bigcup_{T\subset P} A_T \right) \leq \sum_{T\subset P} \Pr(A_T) 
 \leq \frac{n^3}{6} \cdot \frac{6w}{\delta} = \frac{n^3 w}{ \delta}.\]
We obtain that for the
set of $n$ perturbed points $P_x$, the probability that at least 
one triple is flat is upper bound by 
$(n^3w) / (\delta)$.
Thus if the \gridwidth{} $w$ equals
$(p \delta)/(n^3)$ then the probability that at 
least one triple 
is flat is less than $p$.
This, together with Lemma~\ref{lemma:flat} 
shows a lower bound on the required \gridwidth{}.
Due to the correspondence of the norm, the \gridwidth{} and 
the bit-complexity, all three claims of the theorem are proven.
\end{proof}

\section{Bounding Expected Costs}
\label{sec:exp}

We introduced the concept of a triple of points being $w$-flat and we showed that if no triple of points in $P_x$ is $w$-flat then we can snap $P_x$ to a grid $\Gamma$ with width $w$. We used this concept to prove that with high probability, a perturbed point set $P_x$ has a norm which is polynomial in $n$ (and therefore, its order type can be represented by a point set that uses a logarithmic number of bits each).
However, this in itself does not say anything about 
the \emph{expected} costs. It could be that with 
low probability, the costs of $P_x$ is very high.
We continue to prove Theorem~\ref{thm:Expected} and in order to upper bound probabilities, we need a standard integration trick which we state below. For completeness, we provide the proof for the lemma in de appendix.

\begin{restatable}{lemma}{tonelli}
    \label{lem:Tonelli}
    Given a function 
    $f: \Omega \rightarrow \{1,\ldots,b\}$ 
    and assume that $\Pr(f(x)>b) = 0$.
    Then it holds that
    \[\E[f] = \sum_{z=1}^{b} z\Pr(f(x) = z)  \ = \ \sum_{z=1}^{b} \Pr(f(x)\geq z).\]
\end{restatable}

\Expected*

\begin{proof}

We start by proving~\ref{itm:EXPWIDTH} and lower bounding the grid width $w$.
Note that $w$ is a continuous variable and thus, the expected
value is defined via integrals.
The grid width is per definition upper bounded by $1$ so we
can write the expected \gridwidth{} as:
\begin{align*}
    \E[w(P)] &= \int_{x\in \Omega_\delta} w(P_x)\Pr(x) \ dx \\
    &=\int_0^1 z \cdot \Pr(w(P) = z)\ dz.
\end{align*}  
    Observe that we suppress the $x \in \Omega_\delta$ in our notation
    and the second line is hiding the underlying probability space $\Omega_\delta$. However whenever we speak about probabilities,
    it is with respect to $\Omega_\delta$.
    Let us denote $l = \delta / 2 n^3$.
    By Theorem~\ref{thm:WHP} it holds that $\Pr(w(P) \geq l) \geq \frac{1}{2}$. 
    It follows that:
\begin{align*}    
    \E[w(P)] &\geq \int_l^1 z \cdot \Pr(w(P) = z)\ dz \\
             &\geq \int_l^1 l \cdot \Pr(w(P) = z)\ dz \\
             &=  l \ \Pr(w(P) \geq l) \\
             &=  \frac{\delta}{2n^3} \cdot  1/2 = \frac{\delta}{4 n^3}. \\
\end{align*}

\noindent
This shows the claimed lower bound on the \gridwidth{}. 
\vspace{10px}

Now we are turning to proving~\ref{itm:EXPNORM}  by upper bounding the norm.
Recall that the expected value of the norm is given by 
$\E[\nu(P)] = \sum_{z=1}^{\infty} z \Pr(\nu(P) = z)$. 
However order types of points in general position
have a norm of at most $L = 2^{2^{cn}}$.
The point set $P_x$ has collinearities with 
probability $0$ and thus, we are neglecting them.
The expected value of our norm is 
$\E[\nu(P)] = \sum_{z=1}^{L} z \Pr(\nu(P) = z)$. 
Using Lemma~\ref{lem:Tonelli} we get:
\begin{align*}
\mathbb{E}[\nu(P)] &= \sum_{z=1}^{L} \Pr(\nu(P) \ge z).
\end{align*}
\noindent
A consequence of Theorem~\ref{thm:WHP} is that with a magnitude of perturbation of $\delta$, the probability that the norm is a value greater than $z$ is at most $\frac{2 n^3}{z \delta}$. Also, any probability is at most~$1$. 
We denote $l = \lfloor\frac{2n^3}{\delta}\rfloor$.
We upper bound the expected value of the norm as follows:

\begin{align*}
\mathbb{E}[\nu(P)] 
&=  \sum_{z=1}^{ l }\Pr(\nu(P) \ge z)  \ + 
\ \sum_{z=l+1}^{L} \Pr(\nu(P) \ge z)   \\
 &\le  \sum_{z=1}^{l} 1 \  + \
 \sum_{z=l+1}^{L} \frac{2n^3}{z \delta}  
\end{align*}
The Harmonic series $H_K = \sum_{i=1,\ldots, K}1/i$ is upper
bounded by $\log_e K +1$ and lower bounded by $\log_e (K+1)$.
This gives
\begin{align*}
&\le  l + \frac{2n^3}{\delta}\left(\, \log_e L + 1
- \log_e(l+1)\, \right).
\end{align*}
We observe that
\[ \log_e(L) = \log_e\left( 2^{2^{cn}}\right) \leq \log_e 2\cdot  2^{cn} \leq 2^{cn -1}.\]
Using the assumption $n\geq 2$ and $\delta\leq 1/2$ in conjuction with the fact that $l = \lfloor\frac{2n^3}{\delta}\rfloor$ we see:
\[\log_e(l+1) \geq \log_e\left( \frac{2\cdot 8}{1/2} + 1 \right)  \geq 1.\]

Substituting this value gives us an upper bound for the expected norm:
\begin{align*}
\mathbb{E}[\nu(P)]  &\le  \frac{2n^3}{\delta} (1 + 2^{cn-1} - 1)\le  \frac{ n^3}{\delta} 2^{cn}.
\end{align*}
This finishes the proof of the second part of the theorem.
\vspace{10px}

We are now turning to part~\ref{itm:EXPBITS} of the theorem
and look at the bit-complexity.
In the previous section we used the fact that the bit-complexity 
of each coordinate is upper bound by $\log(\nu(P))$.
For analyzing an upper bound in Theorem~\ref{thm:WHP}, this rough upper bound sufficed.
In this section however, we use a more refined analysis.
For this, consider a slightly different snapping algorithm.
Let $p_x$ be a point of $P_x$.
We snap $p_x$ onto a point $p'$ using as few coordinate bits as possible.
To be more precise, we define $\BITC(p_x)$ as
the minimum $k$ such that there is no $2^{-k}$-flat
triple in $P_x$ involving $p_x$. Lemma~\ref{lemma:flat} then guarantees that all triples that $p_x$ is a part of maintain their order type when snapping to coordinates that use $k$ bits.

We can observe the following lemma, with an analysis which is analogue 
to the union-bound analysis given in the proof of Theorem~\ref{thm:WHP}. 
The proof of this lemma is deferred to the appendix.
\begin{restatable}{lemma}{boetsNEW}
    \label{lem:bits}
    For all values $k \geq 1$ it holds that:
    \[\Pr(\BITC(p_x) \geq k) \leq \left(\frac{2n^2}{\delta} \right) 2^{-k}.\]
\end{restatable}

 Using Lemma~\ref{lem:Tonelli}, the expected value of $\BITC(p_x)$ can be expressed as:
\[
\mathbb{E}(\BITC(p_x)) = \sum_{k=1}^\infty k \Pr(\BITC(p_x) = k) \  
= \sum_{k=1}^\infty \  \Pr(\BITC(p_x) \geq k) .
\]
We split the sum, with the splitting point being $l = \lceil \log(2n^2 / \delta) \rceil$.
\[
\mathbb{E}(\BITC(p_x)) = \sum_{k=1}^{l} \Pr(\BITC(p_x) \geq k) \  +
\sum_{k=l+1}^{\infty} \Pr(\BITC(p_x) \geq k).
\]
Now we note that any probability is at most~$1$, and that 
the probability within the right sum can be upper bound by 
$\frac{2n^2}{\delta2^{k}}$ according to Lemma~\ref{lem:bits}. Thus we obtain:
\[
\mathbb{E}(\BITC(p_x)) \le \sum_{k=1}^{l} 1 \  + \left(\frac{2n^2}{\delta} \right) 
\sum_{k=l+1}^{\infty} 2^{-k}.
\]
Observe that $\sum_{k=l+1}^{\infty} 2^{-k}  = 2^{-l}$ and that $l = \lceil \log(2n^2 / \delta) \rceil$ so we get:
\[
\mathbb{E}(\BITC(p_x)) \le l + \left(\frac{2n^2}{\delta} \right) 2^{-l} 
\leq \lceil \log(2n^2 /  \delta) \rceil + \left(\frac{2n^2}{\delta} \right) \left( \frac{\delta}{2n^2} \right)  = 
\lceil \log(n^2 / \delta) \rceil + 2 .
\]
This finishes the proof.
\end{proof}

\noindent
\section{Conclusion}
\label{sec:conclusion}

We studied the realizability of order types under \SmoothAna{}.
The input for our analysis is a worst-case, real-valued, planar, ordered point set $P$ 
subject to a uniform perturbation of magnitude $\delta$ and we analyze the perturbed 
point set $P_x$ using a simple snapping algorithm and using three cost measures: 
(1) minimal grid width for order type preserving snapping, 
(2) the norm of $P_x$ and (3) the number of bits needed per coordinate to represent the order type of $P_x$. 
Theorem~\ref{thm:WHP} shows that with high probability, the cost measures 
scale proportional to $n^3$ and inversely proportional to the 
magnitude of perturbation $\delta$. Which means that with high probability 
the perturbed point set $P_x$ is polynomial in $\delta$ 
and $n$ as opposed to doubly exponential in $n$.
Theorem~\ref{thm:Expected} extends these results and shows that 
also in smoothed expectancy the order type of the perturbed point set $P_x$ has a 
well-behaved realization. Theorem~\ref{thm:Expected} is a stronger 
result than just a statement about the expected cost of a 
random point set, since it shows that even if 
we rapidly push the point set $P_x$ towards a worst-case configuration 
(by lowering $\delta$) the expected cost only worsens at a linear pace. 

Theorem~\ref{thm:Expected} has many theoretical and practical implications. 
For starters, it shows that the decision problem of abstract order type 
realizability can be solved in ``expected \NP-time'': 
if an abstract order type is realizable, then a \NP-algorithm can give a real-valued 
point set $P$ and under smoothed analysis the validity of that point set 
is expected to be polynomial-time verifiable. This makes the 
decision problem of abstract order type realizability the second known
problem which is \ER-complete, and solvable in ``expected \NP-time''.
Due to Mn\"{e}v's universality theorem, we know that polynomials and
order types are closely linked. The current state of the art to solve
polynomial equations and inequalities has high running time and does not scale
to solve practical instances precisely. 
On the other hand, solving integer programs and satisfiability
can be done very fast for large instances in practice.
These observations together with the results in this paper raise the question:

\vspace{0.3cm}
Can we solve arbitrary polynomial equations in ``expected \NP-time''?
\vspace{0.3cm}

One of the first challenges to tackle this question is
to find the right model of perturbation to make this a 
mathematically precise question.

The practical implication of Theorem~\ref{thm:Expected} is that real-valued point 
sets are expected under smoothed analysis to maintain their 
combinatorial properties when you only read their coordinates 
to finite polynomial precision. This justifies the use of word-RAM 
computations for algorithms that want to compute the 
convex hull or the number of triangulations of point sets in practice.

Lastly we want to reiterate that Theorem~\ref{thm:WHP} is a generalization of the 
recent results by Fabila-Monroy and Huemer~\cite{RandomPoints} and that 
Theorem~\ref{thm:WHP} together with Theorem~\ref{thm:Expected} (3) is a generalization of the 
recent results by Devillers, Duchon, Glisse, and  Goaoc~\cite{xavier}. 
We consider the strength of our results to not only be their generality, but also their simplicity. 
Even though the simplicity of a technique is a subjective criterion, 
we hope that this new approach to probabilistic order type realizability 
helps progress future work in this active research field.


\newpage
\bibliographystyle{plainurl}
\bibliography{references}


\newpage
\appendix

\begin{figure}[htbp]
    \centering
    \includegraphics{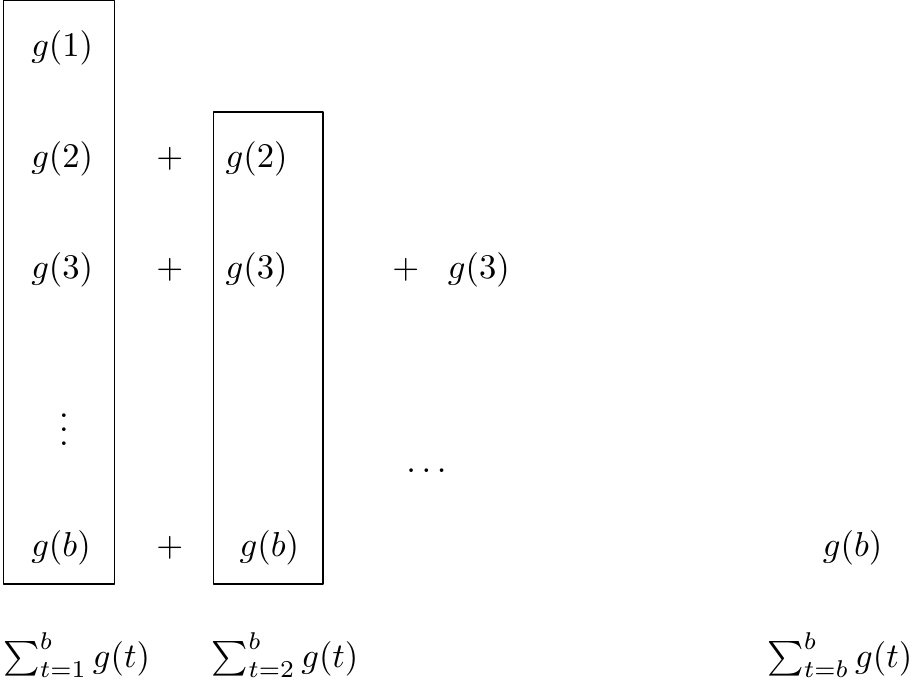}
    \caption{Rewriting the sums.}
    \label{fig:Sums}
\end{figure}

\section{Poof of Lemma~\ref{lem:Tonelli}}
\tonelli*

\begin{proof}
    To simplify notation, we write $g(z) = \Pr(f(x) = z)$.
    We can now write the expectation as
    \begin{align*}
        \E(f) &= \sum_{z=1}^{b} z \Pr(f(x)=z) = \sum_{z=1}^{b} z g(z) = \sum_{z=1}^{b} \sum_{t=1}^{z}g(z) \\
    \end{align*}
    For the next step, we refer to Figure~\ref{fig:Sums}, for an illustration.
    \begin{align*}
         &= \sum_{z=1}^{b}  \sum_{j= z}^{b} g(j) = \sum_{z=1}^{b}   \Pr(f(x) \geq z)
    \end{align*}
    Note that this works also for the special case $b = \infty$.
\end{proof}

\section{Proof of Lemma~\ref{lem:bits}}

\boetsNEW*

\begin{proof}
    Let us denote $w=2^{-k}$ and $p = \left(\frac{2n^2}{\delta} \right) 2^{-k}$. 
    We have to show that at least one triple involving 
    $p_x$ is $w$-flat with probability at most $p$.
    There are at most $n^2/2$ triples involving $p_x$.
    Those triples define $3n^2/2$ lines. 
    It suffices to upper bound the probability that the
    third point of a triple is within distance 
    $w/\sqrt{2}$ to the line through the other two points.
    The probability for a point to be within distance $w/\sqrt{2}$ of a line $\ell$
    is at most $w/\delta$, as explained in  Lemma~\ref{lemma:flat}.
    Thus using the union-boud on all those events we get
    \[\Pr(\BITC(p_x) \geq k) \leq \frac{3n^2w}{2\delta}   
    < \left(\frac{2n^2}{\delta} \right) 2^{-k} .\]
    This shows the claim.
\end{proof}

\end{document}